\documentclass[a4paper,10pt]{article}
\usepackage[utf8]{inputenc}

\usepackage{amsmath}
\usepackage{amsthm}
\usepackage{amssymb}
\usepackage{paralist}
\usepackage{tikz}

\newtheorem{theorem}{Theorem}[section]
\newtheorem{lemma}[theorem]{Lemma}
\newtheorem{corollary}[theorem]{Corollary}

\numberwithin{equation}{section}

\title{A short and elegant proof of a theorem of J.-E. Pin}

\author{Michiel de Bondt}

\begin{document}

\maketitle

\begin{abstract}
We shorten the proof of a theorem of J.-E. Pin (theorem \ref{corank3} below), 
which can be found in his thesis. The part of the proof which is my own (not Pin's) is a complete replacement of the same part in an earlier version of this paper.
\end{abstract}

\section{Introduction}

In the literature, there are several references to a theorem which appeared in the thesis of J.-E. Pin (theorem \ref{corank3} below). Pin promoted in 1981, but the theorem has not been published since then. Nevertheless, the parts which were relevant for the theorem were shared with the author by Pin. I checked the proof and found no real errors. In the proof which is presented in this paper, I follow the proof of Pin, but only up to a certain point. At that point, I take another road, which substantially shortens the proof.

In the previous version of this paper, I took another road at the same point as well. But the road in this paper is yet another road, because I completely replaced my part of the proof of Pin's theorem in this version of the paper. The proof has become even shorter in this paper, because there is no deep case selection on the newest road. The road in the previous version came down to an analysis of four deep cases. The four cases were solved without doing anything fancy, except maybe the last case, which was the hardest case.

The first step in the transition from my old road to this road was a merge of the first three cases. With that, the construction of the second case was adapted. The unifying argument of the three merged cases came down to an easy subcase of lemma \ref{pq}, namely the case $q = 3x$ and $p \ne 2x$ of it. The second step was making a lemma of the argument and include the case $q = 3x$ and $p = 2x$, because that provided the proof of the fourth case. So lemma \ref{pq} was born and all cases were merged to form lemma \ref{reach}. The last step was dropping the condition $q = 3x$ in lemma \ref{pq}: a generalization which does not serve the unified proof of the four cases.

In the previous version of this paper, I criticized Pin for having a layered proof on the road he took, as opposed to my road in the previous version. This is because both Pin's road and my old road involve argumentations for four deep cases, where Pin's argumentations seemed more complicated. But surprisingly, the road in this version does have a layering structure as well.

But before we can formulate Pin's theorem and its proof, we first need to make some definitions and fix some notations. A \emph{deterministic finite automaton} (DFA) is a tuple $(Q,\Sigma,q_0,F)$, where $Q$ is a finite set, each $s \in \Sigma$ is a function from $Q$ to itself, $q_0 \in Q$, and $F \subseteq Q$. The parameters $q_0$ and $F$ are irrelevant in this paper, so we will ignore them. Following a common practice, we use postfix notation for each $s \in \Sigma$, so we write $qs = q'$ instead of $s(q) = q'$. Furthermore, if $R \subseteq Q$, then we write $Rs$ for $\bigcup_{q \in R} qs$.

We denote by $\Sigma^{*}$ the set of all words over $\Sigma$. We follow a usual
formal definition of the set of words as an inductive type, which comes down
to the following. A word $w$ is either the empty word $\lambda$ or $sw'$, where
$s$ is the starting letter and $w'$ is the rest of $w$ as another word. 
Now $Rw$ and $|w|$ can be defined by induction on the structure of $w$ as follows.
\begin{align*}
R\lambda &= R & |\lambda| &= 0\\
R(sw') &= (Rs)w' & |sw'| &= 1 + |w'|
\end{align*}
We call $|w|$ the \emph{length of $w$}, and denote by $\Sigma^i$ the set of all words of length $i$ over $\Sigma$. We denote by $|R|$ the size of a subset $R \subseteq Q$. The following theorem can be found in the thesis of J.-E. Pin.

\begin{theorem}[J.-E. Pin] \label{corank3}
Let $c \le 3$. Suppose that $(Q,\Sigma,\ldots)$ is a DFA, and there exist 
$w \in \Sigma^{*}$ such that $|Qw| \le |Q| - c$. Then we can choose $w$ 
such that $|w| \le c^2$.
\end{theorem}

If we omit the condition that $c \le 3$ in theorem \ref{corank3}, then
we obtain Pin's corank conjecture. Pin's corank conjecture does not hold for
$c = 4$, see \cite{K01} for $|Q| = 6$ and \cite{10.1007/978-3-319-40946-7_15}
for $|Q| \ge 6$.

\begin{corollary}[J.-E. Pin] \label{cerny4}
Suppose that $(Q,\Sigma,\ldots)$ is a DFA, and there exist $w \in \Sigma^{*}$
such that $|Qw| = 1$. If $n = |Q| \le 4$, then we can choose $w$ such that 
$|w| \le (n-1)^2$.
\end{corollary}

\begin{proof}[Proof (J.-E. Pin)]
Take $c = n - 1$ in theorem \ref{corank3}.
\end{proof}

If we omit the condition that $|Q| \le 4$ in corollary \ref{cerny4},
then we obtain \v{C}ern\'y's conjecture, see \cite{C64} and \cite{volkov}.
\v{C}ern\'y's conjecture plays a central role in the theory of synchronizable
DFAs. \v{C}ern\'y's conjecture has been proved by way of exhaustive search 
for $|Q| = 5$ \cite{10.1007/978-3-319-40946-7_15}, $|Q| = 6$ \cite{BDZ17}, and 
$|Q| = 7$ \cite{BDZ18}.

Pin uses the following theorem in his proof of theorem \ref{corank3}. He does that by way of the case $c = 3$ of corollary \ref{pinappl} below.

\begin{theorem}[J.-E. Pin] \label{pin}
Suppose that $(Q,\Sigma,\ldots)$ is a DFA, and there exist $u \in \Sigma^{*}$
such that $|Qu| \le |Q| - c$. Take any $w \in \Sigma^{*}$ such that 
$|Qw| \le |Q| - c + 1$. Then there exist $m \in \Sigma^{*}$ such that 
$|Qwmw| \le |Q| - c$ and $|m| \le c$.
\end{theorem}

Theorem \ref{pin} is slightly more explicit than Proposition 5 in \cite{pinlinalg}.

\begin{corollary} \label{pinappl}
Suppose that $(Q,\Sigma,\ldots)$ be a DFA, and there exist $u \in \Sigma^{*}$ such that $|Qu| \le |Q| - c$. Take $w \in \Sigma^{*}$, such that $|Qw| \le |Q| - c + 1$. If $|w| \le (c-1)c/2$, then we can choose $u$ of the form $u = wm$ such that $|u| \le c^2$.
\end{corollary}

\begin{proof}
We can use theorem \ref{pin} (or theorem \ref{franklpin} below), to infer that there exist $m \in \Sigma^{*}$, such that $|Qwm| \le |Q| - c$ and $|m| \le c(c+1)/2$. Hence $|wm| \le c^2$ if $|w| \le (c-1)c/2$.
\end{proof}

Using corollary \ref{pinappl}, it is possible to prove the case $c \le 2$ of theorem \ref{corank3} by induction on $c$. But just as Pin, we will only use the case $c = 3$ of corollary \ref{pinappl} in the proof of theorem \ref{corank3}. For the subcase of the case $c = 3$ of theorem \ref{corank3} which is not covered by corollary \ref{pinappl}, we give another proof. 

The following result was not known yet when Pin wrote his thesis. 

\begin{theorem}[P. Frankl and J.-E. Pin] \label{franklpin}
Suppose that $(Q,\Sigma,\ldots)$ is a DFA, and there exist $u \in \Sigma^{*}$
such that $|Qu| \le |Q| - c$. Take $R \subseteq Q$ such that 
$|R| \le |Q| - c + 1$. Then there exist $m \in \Sigma^{*}$ such that 
$|Rm| \le |Q| - c$ and $|m| \le c(c+1)/2$.
\end{theorem}

Theorem \ref{franklpin} follows from a conjecture and Proposition 3.1 in 
\cite{pintwoprobl}.
The conjecture is that $p(s,t) = \binom{s+t}{t}$, where $p(s,t)$ is defined 
by a property in \cite{pintwoprobl}, and it has been proved by P. Frankl in 
\cite{frankl}.

\begin{corollary}[J.-E. Pin] \label{franklpincor}
Suppose that $(Q,\Sigma,\ldots)$ is a DFA, and there exist $u \in \Sigma^{*}$
such that $|Qu| = 1$. If $n = |Q| \ge 4$, then we can choose $u$ such that 
$|u| \le (n^3-n)/6 - 1$.
\end{corollary}

\begin{proof}[Proof (J.-E. Pin)]
Suppose that $n \ge 4$. From theorem \ref{corank3}, it follows that there exist $w \in \Sigma^{*}$ such that $|Qw| \le n - 3$ and $|w| \le 9$. By applying theorem 
\ref{franklpin} for $c = 4, 5, \ldots, n-1$, in that order, the reader can prove 
that there exist $v \in \Sigma^{*}$, such that $|Qwv| = 1$ and 
$$
|v| \le \sum_{c=4}^{n-1} c(c+1)/2 = (n^3-n)/6 - 10
$$
Hence we can take $u = wv$.
\end{proof}

Various authors attribute a weaker version of corollary \ref{franklpincor} to 
J.-E. Pin, namely with $|u| \le (n^3-n)/6 - 1$ replaced by $|u| \le (n^3-n)/6$. 
But this weaker result has never been claimed by J.-E. Pin.

In the last section, we will characterize the situations where the
word $w$ of theorem \ref{corank3} cannot be obtained by trying all words
of greedy type. To define words of greedy type, we first need to define minimal compression words. A word $w$ is a \emph{compression word for $R \subseteq Q$} if $|Rw| < |R|$, and a \emph{minimal compression word for $R$} if $w$ has minimal length as such. A word $u$ is \emph{of greedy type} for $R$ if either $u = \lambda$ or $u = wv$, such that $w$ is a minimal compression word for $R$, and $v$ is of greedy type for $Rw$.

If $c \le 2$ and the conditions of theorem \ref{corank3} are satisfied, then greedy compression will always give you a word which satisfies theorem \ref{corank3}. But this does not hold for $c = 3$. Theorem \ref{greedy} (1) expresses being in the situation where greedy compression cannot give you a word which satisfies theorem \ref{corank3}. With theorem \ref{greedy} (2), (3) and (4), three equivalent characterizations are given. Corollary \ref{pincor} gives a word which does satisfy theorem \ref{corank3} in this situation. In Pin's proof of theorem \ref{corank3}, Pin proves theorem \ref{greedy} (1) $\Rightarrow$ (3) and corollary \ref{pincor}. Pin proves theorem \ref{greedy} (2) $\Rightarrow$ (3) more or less, too.

\section{Proof of theorem \ref{corank3}}

We will essentially follow Pin's proof in this section, starting with the case $c \le 2$. If $c = 0$, then we take for $w$ the empty word $\lambda$. So assume that $1 \le c \le 2$. Then there exists an $x \in \Sigma$ such that $Qx \ne Q$. If $c = 1$, then we take $w = x$. So assume that $c = 2$. We may assume without loss of generality that $1x = 2x$. If $\{1,2\} \subseteq Qx$, then $|Qx^2| \le n-2$. So we may assume without loss of generality that $Qx = Q \setminus \{1\}$. 

There exist $y \in \Sigma$ such that $Qxy \ne Qx$. If $\{1,2\} \subseteq Qxy$ for some $y \in \Sigma$, then $|Qxyx| \le n-2$. So we may assume without loss of generality that $Qxy = Q \setminus \{2\}$ for some $y \in \Sigma$. Since $c \ge 2$, there exist $z \in \Sigma$ such that $Qxyz \notin \{Qx, Qxy\}$. So we may assume without loss of generality that $Qxyz = Q \setminus \{3\}$ for some $z \in \Sigma$. As $\{1,2\} \subseteq Q \setminus \{3\}$, $|Qxyzx| \le n-2$ follows, which proves the case $c = 2$.

Suppose that $Q = \{1,2,\ldots,n\}$, and let $(Q,\Sigma,\ldots)$ be a DFA, such that $|Qw| \ge n-1$ for all $w \in \Sigma^{*}$ such that $|w| \le 3$. Then we can follow the case $c = 2$ of the above proof, to see that the last assumption excludes all cases except the last case. So we can renumber the state set, and choose a word $xyzx$, such that this word gives the following path from $Q$ in the power automaton, where $1 \ne 3x$.
\begin{center}
\begin{tikzpicture}[semithick]
\tikzstyle{nodestyle}=[draw,fill=white,circle,inner sep=0pt,minimum width=1cm]
\fill[lightgray,rounded corners=5mm] (-1,-1) rectangle (1,1); 
\fill[lightgray,rounded corners=5mm] (2,-1) rectangle (8,1); 
\fill[lightgray,rounded corners=5mm] (9,-1) rectangle (11,1); 
\node[nodestyle] (A) at (0,0) {$Q$};
\node[nodestyle] (B) at (3,0) {$\begin{array}{@{}c@{}} Q \setminus {} \\ \{1\} \end{array}$};
\node[nodestyle] (C) at (5,0) {$\begin{array}{@{}c@{}} Q \setminus {} \\ \{2\} \end{array}$};
\node[nodestyle] (D) at (7,0) {$\begin{array}{@{}c@{}} Q \setminus {} \\ \{3\} \end{array}$};
\node[nodestyle] (E) at (10,0) {$\begin{array}{@{}c@{}} Q \setminus {} \\ \{1,3x\} \end{array}$};
\draw[->] (A) edge node[above] {$x$} (B) (B) edge node[above] {$y$} (C) 
(C) edge node[above] {$z$} (D) (D) edge node[above] {$x$} (E);
\end{tikzpicture}
\end{center}
Here, the shaded blocks correspond to cardinalities of the state set. We call a DFA as above, such that $|Qw| \ge n-1$ for all $w \in \Sigma^{*}$ such that $|w| \le 3$, a \emph{Pin DFA}.

\begin{lemma}[J.-E. Pin] \label{pinlemma}
Suppose that $(Q,\Sigma,\ldots)$ is a Pin DFA. Take $s \in \Sigma$. Then $q \mapsto qs$ is either a bijection from $Q \setminus \{1\}$ to $Q \setminus \{1\}$, or a bijection from $Q \setminus \{1\}$ to $Q \setminus \{2\}$. Furthermore, $s$ satisfies one of the following two equations.
\begin{align}
Qs &= Q & 1s &\in \{1,2\} \label{ad} \\
Qs &= Q \setminus \{1\} & 1s &= 2s \label{b1}
\end{align}
In case of \eqref{ad}, $q \mapsto q$ is a bijection from $Q$ to $Q$. In case of \eqref{b1}, $q \mapsto qs$ is a bijection from $Q \setminus \{1\}$ to $Q \setminus \{1\}$, and a bijection from $Q \setminus \{2\}$ to $Q \setminus \{1\}$.
\end{lemma}

\begin{proof} Take $s \in \Sigma$. Suppose first that $Qs = Q$. Then 
$$
Qxs = (Q \setminus \{1\})s = Q \setminus \{1s\}
$$
From $|Qxsx| = n-1$, we deduce that $\{1,2\} \nsubseteq Qxs$, so $1s \in \{1,2\}$. The other claims follow as well.

Suppose next that $Qs \ne Q$. Then $|Qs| = n - 1$. From $|Qsx| = n-1$, we deduce that $\{1,2\} \nsubseteq Qs$, so either $Qs = Q \setminus \{1\}$ or $Qs = Q \setminus \{2\}$. From $|Qszx| = n-1 \ne |Qxyzx|$, we infer that $Qs \ne Qxy = Q \setminus \{2\}$, so 
$Qs = Q \setminus \{1\}$. If we combine $Qs = Q \setminus \{1\}$ with
$$
|(Q \setminus \{1\})s| = |Qxs| = n - 1 = |Qxys| = |(Q \setminus \{2\})s|
$$
then $(Q \setminus \{1\})s = Q \setminus \{1\} = (Q \setminus \{2\})s$ follows. This gives all claims except $1s = 2s$. The last claim subsequently gives $1s = 2s$.
\end{proof}

From $1 \in Q \setminus \{2\} = Qxy$ and $1 \in Q \setminus \{3\} = Qxyz$, it follows that $Qy \ne Q \setminus \{1\}$ and $Qz \ne Q \setminus \{1\}$ respectively. On account of lemma \ref{pinlemma} above,
\begin{align}
Qy &= Q & 1y &= 2 \label{a1} \\
Qz &= Q & 2z &= 3 \label{d1}
\end{align}

We continue with the proof of the case $c = 3$ of theorem \ref{corank3}. If $|Qw| \le |Q| - 2$ for some $w \in \Sigma^{*}$ such that $|w| \le 3$, then theorem \ref{corank3} follows from corollary \ref{pinappl}.  So assume that $|Qw| \ge |Q| - 1$ for all $w \in \Sigma^{*}$ such that $|w| \le 3$. Then up to renumbering states, $(Q,\Sigma,\ldots)$ is a Pin DFA. So the following two lemmas prove this case. We postpone the proof of the first lemma to the next section.

\begin{lemma} \label{pq}
Assume that $(Q,\Sigma,\ldots)$ is a Pin DFA. Let $p,q \in Q$ such that $1 \ne p \ne q \ne 1$. Suppose that
$$
|(Q \setminus \{1,p\})m| = n-2 = |(Q \setminus \{1,q\})m| 
$$
for all $m \in \Sigma^{*}$ such that $|m| \le 4$. Then $|Qw| \ge n-2$ for all $w \in \Sigma^{*}$.
\end{lemma}

\begin{lemma} \label{reach}
Assume that $(Q,\Sigma,\ldots)$ is a Pin DFA. Then there exist $p \in Q$ with $1 \ne p \ne 3x$, and $w \in \Sigma^5$, such that $Qw = Q \setminus \{1,p\}$.
\end{lemma}

\begin{proof}
From $(Q \setminus \{1\}) x = Q \setminus \{1\}$, we infer that
$$
Qxyzx^2 = (Q \setminus \{1,3x\}) x = Q \setminus \{1,3x^2\}
$$
So if $3x^2 \ne 3x$, then we can take $p = 3x^2$. So assume that $3x^2 = 3x$. From $(Q \setminus \{1\})x = Q \setminus \{1\}$ and $3x \ne 1 \ne 3$, we infer that $3x = 3$. We distinguish two cases.
\begin{itemize}

\item $1z = 1$ or $3z = 1$.

As $1z \ne 3z$, there exists a $p \in Q \setminus \{1\}$ such that $\{1,3\}z = \{1,p\}$. From $2z = 3$ and $1z \ne 2z \ne 3z$, $p \ne 3$ follows. As $3x = 3$ and $Qz = Q$,
$$
Qxyzxz = (Q \setminus \{1,3\}) z = Q \setminus \{1,p\}
$$
and the conclusion follows.

\item $1z \ne 1$ and $3z \ne 1$.

As $1z \in \{1,2\}$, $1z = 2$ follows. Furthermore, $3z \notin \{1,1z,2z\} = \{1,2,3\}$, so we may assume that $3z = 4$. From $Qz = Q$ and $(Q \setminus \{1\}) x = Q \setminus \{1\}$, we infer that
$$
Qxyz^2x = (Q \setminus \{3\})zx = (Q \setminus \{4\})x = Q \setminus \{1,4x\}
$$
Furthermore, $3x \ne 4x$, so we can take $p = 4x$. \qedhere

\end{itemize}
\end{proof}

\section{Proof of lemma \ref{pq}}

Take $p' = py$ and $q' = qy$.

\begin{lemma} \label{pqlem}
\begin{enumerate}[\upshape(i)]

\item If $s \in \Sigma$ such that $(Q \setminus \{1\})s = Q \setminus \{1\}$, then
$$
\{p,q\}s = \{p,q\}
$$
Furthermore, $\{p,q\} yz = \{1,2\}$.

\item If $s \in \Sigma$ such that $(Q \setminus \{1\})s = Q \setminus \{2\}$, then
$$
\{p,q\}s = \{p',q'\}
$$
Furthermore, $\{p',q'\} z = \{1,2\}$.

\item $1 \in \{p',q'\}$ and $2 \in \{p,q\}$.

\end{enumerate}
\end{lemma}

\begin{proof}
\begin{enumerate}[(i)]

\item Take $s \in \Sigma \cup \{\lambda\}$ such that $(Q \setminus \{1\})s = Q \setminus \{1\}$. As $Qyz = Q$, there exists a pair $\{d,e\} \subseteq Q$, such that $\{d,e\}yz = \{1,2\}$. As $1yz = 2z = 3 \notin \{1,2\}$, $1 \notin \{d,e\}$ follows. So there exists a pair $\{a,b\} \subseteq Q \setminus \{1\}$, such that 
$$
\{a,b\}s = \{d,e\}
$$
As $|\{a,b\}m| = 1$ for some $m \in \Sigma^{*}$ such that $|m| \le 4$, $\{a,b\} \nsubseteq Q \setminus \{1,p\}$ follows. Similarly, $\{a,b\} \nsubseteq Q \setminus \{1,q\}$ follows. This is only possible if $\{a,b\} = \{p,q\}$. Taking $s = \lambda$ gives $\{p,q\} = \{p,q\}\lambda = \{d,e\}$. So $\{p,q\}s = \{p,q\}$.

\item Take $s \in \Sigma$ such that $(Q \setminus \{1\})s = Q \setminus \{2\}$. As $Qz = Q$, there exists a pair $\{d',e'\} \subseteq Q$, such that $\{d',e'\}z = \{1,2\}$. As $2z = 3 \notin \{1,2\}$, $2 \notin \{d',e'\}$ follows. So there exists a pair $\{a,b\} \subseteq Q \setminus \{1\}$, such that
$$
\{a,b\}s = \{d',e'\}
$$
Just as in the proof of (i), $\{a,b\} = \{p,q\}$ can be obtained. Taking $s = y$ gives $\{p',q'\} = \{p,q\}y = \{d',e'\}$. So $\{p,q\}s = \{p',q'\}$.

\item As $Qz = Q$, we infer from $\{p',q'\}z = \{1,2\}$ and $1z \in \{1,2\}$ that $1 \in \{p',q'\}$. 

Suppose that $2 \notin \{p,q\}$. As $(Q \setminus \{1\})y = Q \setminus \{2\}$, there exists a pair $\{a,b\} \subseteq Q \setminus \{1\}$, such that
$$
\{a,b\}y = \{p,q\}
$$
From $\{p,q\}yz = \{1,2\}$, it follows that $\{a,b\} = \{p,q\}$ can be obtained as before. So $1 \in \{p',q'\} = \{p,q\}y = \{p,q\}$. Contradiction, so $2 \in \{p,q\}$. \qedhere

\end{enumerate}
\end{proof}

So we have the following transition graph of the power automaton.
\begin{center}
\begin{tikzpicture}[semithick]
\tikzstyle{nodestyle}=[draw,fill=white,circle,inner sep=0pt,minimum width=12mm]
\fill[lightgray,rounded corners=5mm] (-2,-1) rectangle (5,1); 
\node[nodestyle] (de) at (0,0) {$\{p,q\}$};
\node[nodestyle] (fg) at (2,0) {$\{p',q'\}$};
\node[nodestyle] (hi) at (4,0) {$\{1,2\}$};
\draw[->] (de) edge[out=-150,in=150,looseness=5] node[left] {$x'$} (de);
\draw[->] (de) edge node[above] {$y'$} (fg) (fg) edge node[above] {$z$} (hi);
\draw (hi) edge ++(1.2,0);
\node at (7.5,0) {$\begin{aligned} (Q \setminus \{1\}) x' &= Q \setminus \{1\} \\
(Q \setminus \{1\}) y' &= Q \setminus \{2\} \end{aligned}$};
\end{tikzpicture}
\end{center}

\begin{lemma} \label{pqeq}
\begin{align*} 
\{p,q\} &= \{2,3\} & \{p',q'\} = \{1,3\}
\end{align*}
\end{lemma}

\begin{proof}
From lemma \ref{pqlem} (iii), it follows that we can take $e \in Q$ and $c' \in Q$, such that  $\{p,q\} = \{2,e\}$ and $\{p',q'\} = \{1,c'\}$. We must show that $e = 3 = c'$. For that purpose, we distinguish two cases.
\begin{itemize}

\item $1z = 1$.

Then lemma \ref{pqlem} (i) with $s = z$ gives $\{p,q\}z = \{p,q\}$, so $\{2,e\}z = \{2,e\}$. As $2z = 3$, $e = 3$ follows. Furthermore, $ez = 2$ is obtained. As $e = 3$ and $Qz = Q$, it suffices to show that $c'z = 2$ as well. This follows from $\{1,c'\}z = \{p',q'\} z = \{1,2\}$ and $1z = 1$.

\item $1z = 2$.

Then lemma \ref{pqlem} (ii) with $s = z$ gives $\{p,q\}z = \{p',q'\}$, so $\{2,e\}z = \{1,c'\}$. As $2z = 3$, $c' = 3$ follows. Furthermore, $ez = 1$ is obtained. As $c' = 3$ and $Qz = Q$, it suffices to show that $c'z = 1$ as well. This follows from $\{1,c'\}z = \{p',q'\} z = \{1,2\}$ and $1z = 2$.
\qedhere

\end{itemize}
\end{proof}

Now the following lemma proves lemma \ref{pq}.

\begin{lemma}
For all $s \in \Sigma$,
\begin{align*}
\{1,2,3\}s &\subseteq \{1,2,3\} & (Q \setminus \{1,2,3\}) s &= Q \setminus \{1,2,3\}
\end{align*}
\end{lemma}

\begin{proof}
The first equality follows from the second equality if $Qs = Q$. So we only need to prove the first equality if $1s = 2s$. We distinguish two cases.
\begin{itemize}

\item $(Q \setminus \{1\})s = Q \setminus \{1\}$. 

From lemma \ref{pqlem} (i), it follows that $\{p,q\}s = \{p,q\}$, so $\{2,3\}s = \{2,3\}$ on account of lemma \ref{pqeq}. Hence
$$
(Q \setminus \{1,2,3\})s = \big((Q \setminus \{1\}) \setminus \{2,3\} \big)s = (Q \setminus \{1\}) \setminus \{2,3\} = Q \setminus \{1,2,3\}
$$
Furthermore, $\{1,2,3\}s = \{2,3\}s = \{2,3\}$ if $1s = 2s$.

\item $(Q \setminus \{1\})s = Q \setminus \{2\}$. 

From lemma \ref{pqlem} (ii), it follows that $\{p,q\}s = \{p',q'\}$, so $\{2,3\}s = \{1,3\}$ on account of lemma \ref{pqeq}. Hence
$$
(Q \setminus \{1,2,3\})s = \big((Q \setminus \{1\}) \setminus \{2,3\} \big)s = (Q \setminus \{2\}) \setminus \{1,3\} = Q \setminus \{1,2,3\}
$$
Furthermore, $\{1,2,3\}s = \{2,3\}s = \{1,3\}$ if $1s = 2s$. \qedhere

\end{itemize}
\end{proof}

The state names $d$ and $e$ were taken from Pin's thesis. The same holds for the claims 
$1 \notin \{d,e\}$, $2 \notin \{d',e'\}$ and $1 \in \{d',e'\}$. On account of the last claim, we can write $\{d',e'\} = \{1,c'\}$. Pin used the following state naming in his thesis.
\begin{center}
\begin{tikzpicture}[x=6mm,y=6mm]
\fill[lightgray!50] (-3.4,0) rectangle (0,2);
\foreach[count = \n] \s in {1,2,3,d,e,c'} {
  \node[anchor=base] at (\n-0.5,1.3) {$\s$};
  \draw (\n-1,1) -- ++ (0,1);
}  
\foreach[count = \n] \s in {a,b,s,d,e,c} {
  \node[anchor=base] at (\n-0.5,0.3) {$\s$};
  \draw (\n-1,0) -- ++ (0,1);
}
\node[anchor=base east] at (0,1.3) {This paper};
\node[anchor=base east] at (0,0.3) {Pin's thesis};
\draw (-3.4,1) -- (6,1) (-3.4,0) rectangle (6,2);
\node[anchor=base] at (9.3,1.3) {$\{d',e'\} = \{1,c'\}$};
\end{tikzpicture}
\end{center}

Notice that the symbols $x, y, z$ do not suffice for the word $w$ of theorem \ref{corank3}, if and only if $\{1,2,3\} x = \{2,3\}$ and $\{1,2,3\}y = \{1,2,3\} = \{1,2,3\}z$. In that case, an extra symbol $s$ such that $\{1,2,3\} s \nsubseteq \{1,2,3\}$ is necessary, but only on one spot in the word $w$ of theorem \ref{corank3}. Symbol $s$ will provide a transition away from some set $R \in \big\{Q \setminus \{1,3\},Q \setminus \{1,2\}\big\}$ as follows. If $(Q \setminus \{1\}) s = Q \setminus \{1\}$, then $Rs$ contains $\{2,3\}$. If $(Q \setminus \{1\}) s = Q \setminus \{2\}$, then $Rs$ contains $\{1,3\}$. 

If $1s = 1$ and $2s = 2$, then $s$ cannot replace one of the symbols $x,y,z$, even if we allow renumbering of the states. So four symbols may be necessary for the word $w$ of theorem \ref{corank3}. Conversely, the necessarity of four symbols does \emph{not} imply that $2s = 2$, which the reader may show. If four symbols are necessary, then $3y = 3$ and $1z = 1$, because we cannot choose $y$ and $z$ to be the same symbol.

\section{Theorem \ref{corank3} and greedy compression}

Let $(Q,\Sigma,\ldots)$ be a DFA, and suppose that there exist $w \in \Sigma$, such
that $|Qw| \le n - 3$. Then we can choose $w$ of greedy type, which we do.
From theorem \ref{franklpin}, it follows that there exists an $i \le 3$, 
such that
\begin{align*}
|Qw_1| = |Qw_1w_2| = \cdots = |Qw_1w_2 \cdots w_i| &= n - 1 \\
\intertext{and a $j \le 6$, such that}
|Qw_1w_2 \cdots w_{i+1}| = |Qw_1w_2 \cdots w_{i+2}| = \cdots =
|Qw_1w_2 \cdots w_{i+j}| &= n - 2
\end{align*}
Consequently,
\begin{enumerate}

\item[(i)] If $|w| \ge 4$, then $|Qw_1w_2w_3w_4| \le n - 2$;

\item[(ii)] If $|w| \ge 10$, then $|Qw_1w_2 \cdots w_{10}| \le n - 3$.

\end{enumerate}
Due to (i) above, theorem \ref{greedy} (1) below expresses that $w$ cannot 
be chosen of greedy type.

\begin{theorem} \label{greedy}
Suppose that $(Q,\Sigma,\ldots)$ is a DFA, such that $|Qw| = n - 3$ for some word 
$w$ of length at most $9$. Then the following statements are equivalent:
\begin{enumerate}[\upshape (1)]

\item For every choice of $w$, we have $|w| \ge 4$ and $|Qw_1w_2w_3w_4| \ge n - 1$;

\item There is a numbering of the states, such that $(Q,\Sigma,\ldots)$ is a Pin DFA,
such that 
$$
|(Q \setminus \{1,3x\})m| = n-2
$$
for all $m \in \Sigma^{*}$ with $|m| \le 5$. Furthermore, the following holds if $3x = 3$, $3zx = 3z$ and $3z^2 = 1$: there does \emph{not} exist a symbol $s \in \Sigma$ such that
\begin{align*}
Qs &= Q & 1s &= 1 & 2s &= 3z & 3s &= 3 & (3z)s &= 2 
\end{align*}

\item There is a numbering of the states, such that every $s \in \Sigma$ satisfies one of the following three equations:
\begin{align}
Qs &= Q & 1s &= 1 & 2s &= 2 & 3s &= 3 & 4s &= 4 \label{i} \\
Qs &= Q & 1s &= 2 & 2s &= 3 & 3s &= 4 & 4s &= 1 \label{a} \\
Qs &= Q \setminus \{1\} \hspace*{-50pt} &&& \hspace*{-50pt} 
1s = 2s && 3s &= 3 & 4s &= 4 \label{b}
\end{align}

\item For every choice of $w$, we have $|w| = 9$ and
\begin{align*}
|Qw_1| = |Qw_1w_2| = |Qw_1w_2w_3| = |Qw_1w_2w_3w_4| &= n - 1 \\
|Qw_1w_2w_3w_4w_5| = \cdots = |Qw_1w_2w_3w_4w_5w_6w_7w_8| &= n - 2
\end{align*}

\end{enumerate}
\end{theorem}

\begin{proof}
Notice that (4) $\Rightarrow$ (1) is trivial. So three implications remain to be proved.

\medskip \noindent
{\bf (1) {\mathversion{bold}$\Rightarrow$} (2)~}
Suppose that (1) holds. The first claim of (2) follows from the proof of the case $c = 3$ of theorem \ref{corank3}. To prove the second claim, suppose contrarily that $3z^2 = 1 = 1s$ and $2s = 3z$ for some $s \in \Sigma$. From $1s = 1$, it follows that $Qs = Q$. From $2z = 3$ and $3z^2 = 1$, we infer that $3z \notin \{1,2,3\}$. So we may assume that $3z = 4$. As $\{1,2\} \subseteq Q \setminus \{4\}$, $|(Q \setminus \{4\})x| \le n - 2$ follows. Consequently,
$$
Qxysx = (Q \setminus \{2\})sx = (Q \setminus \{4\})x = Q \setminus \{1,4x\}
$$
As $(Q \setminus \{1\})x = Q \setminus \{1\}$, lemma \ref{pq} with $p = 3x$ and $q = 4x$ yields a contradiction to (1). So the second claim of (2) holds as well.

\medskip \noindent
{\bf (3) {\mathversion{bold}$\Rightarrow$} (4)~}
Suppose that (3) holds. We can do breadth first search on the power automaton. We obtain that $(Q,\Sigma,\ldots)$ is a Pin DFA, and we can see that $z$ is as $s$ in \eqref{a} and $x$ is as $s$ in \eqref{b}. Furthermore, we may assume that $y = z$. Let $r$ be a wildcard for the states in $Q \setminus \{1,2,3,4\}$. Then we get the following transition graph for the power automaton, where $x$ and $z$ are examples of symbols which perform the transitions they mark, and some transitions are missing.
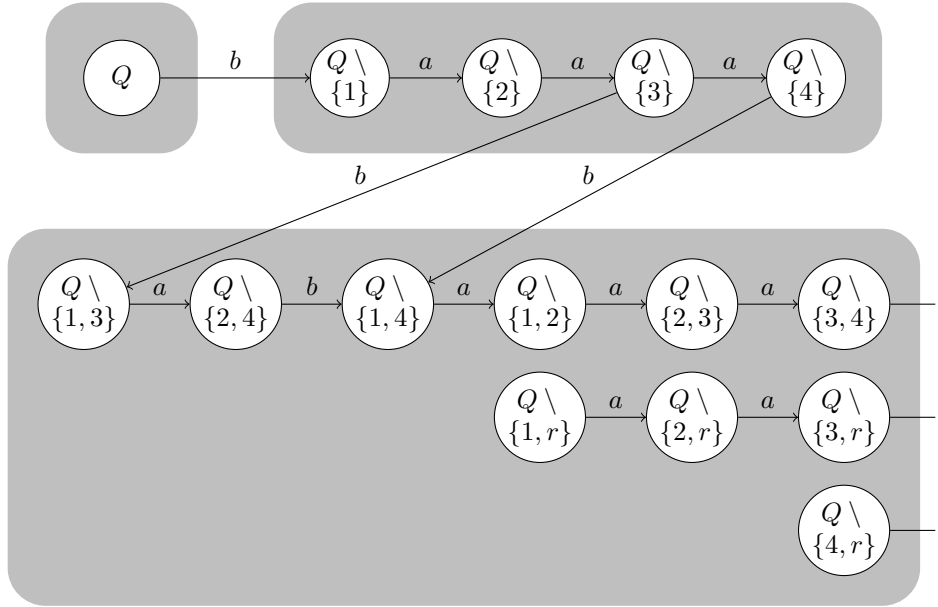
\begin{center}
\begin{tikzpicture}[semithick]
\tikzstyle{nodestyle}=[draw,fill=white,circle,inner sep=0pt,minimum width=1cm]
\fill[lightgray,rounded corners=5mm] (-0.5,2) rectangle (1.5,4); 
\fill[lightgray,rounded corners=5mm] (2.5,2) rectangle (10.5,4);
\fill[lightgray,rounded corners=5mm] (-1,-4) rectangle (11,1); 
\path (-1,-2.5) rectangle (11,1);
\node[nodestyle] (A) at (0.5,3) {$Q$};
\node[nodestyle] (B) at (3.5,3) {$\begin{array}{@{}c@{}} Q \setminus {} \\ \{1\} \end{array}$};
\node[nodestyle] (C) at (5.5,3) {$\begin{array}{@{}c@{}} Q \setminus {} \\ \{2\} \end{array}$};
\node[nodestyle] (D) at (7.5,3) {$\begin{array}{@{}c@{}} Q \setminus {} \\ \{3\} \end{array}$};
\node[nodestyle] (E) at (9.5,3) {$\begin{array}{@{}c@{}} Q \setminus {} \\ \{4\} \end{array}$};
\node[nodestyle] (0) at (0,0) {$\begin{array}{@{}c@{}} Q \setminus {} \\ \{1,3\} \end{array}$};
\node[nodestyle] (1) at (2,0) {$\begin{array}{@{}c@{}} Q \setminus {} \\ \{2,4\} \end{array}$};
\node[nodestyle] (2) at (4,0) {$\begin{array}{@{}c@{}} Q \setminus {} \\ \{1,4\} \end{array}$};
\node[nodestyle] (3) at (6,0) {$\begin{array}{@{}c@{}} Q \setminus {} \\ \{1,2\} \end{array}$};
\node[nodestyle] (4) at (8,0) {$\begin{array}{@{}c@{}} Q \setminus {} \\ \{2,3\} \end{array}$};
\node[nodestyle] (5) at (10,0) {$\begin{array}{@{}c@{}} Q \setminus {} \\ \{3,4\} \end{array}$};
\draw[->] (A) edge node[above] {$x$} (B) (B) edge node[above] {$z$} (C) 
(C) edge node[above] {$z$} (D) (D) edge node[above] {$z$} (E);
\draw[->] (D) edge[in=40,out=-140] node[auto=right,inner sep=2pt] {$x$} (0) 
(E) edge[in=40,out=-140] node[auto=right,inner sep=2pt] {$x$} (2);
\draw[->] (0) edge node[above] {$z$} (1) (1) edge node[above] {$x$} (2) 
(2) edge node[above] {$z$} (3) (3) edge node[above] {$z$} (4) (4) edge node[above] {$z$} (5);
\draw (5) edge ++(1.2,0);
\node[nodestyle] (6) at (6,-1.5) {$\begin{array}{@{}c@{}} Q \setminus {} \\ \{1,r\} \end{array}$};
\node[nodestyle] (7) at (8,-1.5) {$\begin{array}{@{}c@{}} Q \setminus {} \\ \{2,r\} \end{array}$};
\node[nodestyle] (8) at (10,-1.5) {$\begin{array}{@{}c@{}} Q \setminus {} \\ \{3,r\} \end{array}$};
\draw[->] (6) edge node[above] {$z$} (7) (7) edge node[above] {$z$} (8);
\draw (8) edge ++(1.2,0);
\node[nodestyle] (9) at (10,-3) {$\begin{array}{@{}c@{}} Q \setminus {} \\ \{4,r\} \end{array}$};
\draw (9) edge ++(1.2,0);
\end{tikzpicture}
%
\end{center}
Here, the shaded blocks correspond to cardinalities of the state set. The transitions which were omitted are exactly the cardinality-preserving transitions from given state sets, which do not go to the right (their target state sets have \emph{not} been omitted). From this, we infer that (3) is satisfied.

\medskip \noindent
{\bf (2) {\mathversion{bold}$\Rightarrow$} (3)~}
Suppose that (2) holds. From the proof of lemma \ref{reach}, we infer that $3x = 3$ and $1z = 2$. Furthermore, we may assume that $3z = 4$. 

We first show that $s = z$ satisfies \eqref{a} and $s = x$ satisfies \eqref{b}, by providing the missing links $4z = 1$ and $4x = 4$. On account of $Qz = Q$ and $(Q \setminus \{2\})x = Q \setminus \{1\}$,
$$
(Q \setminus \{1,3\})zx = (Q \setminus \{1z,3z\})x = (Q \setminus \{2,4\})x = (Q \setminus \{1,4x\})
$$
Furthermore, $Qxyzx = Q \setminus \{1,3\}$, so
$$
|(Q \setminus \{1,3\})m| = n-2 = |(Q \setminus \{1,4x\})m| 
$$
for all $m \in \Sigma^{*}$ such that $|m| \le 3$. This does not yield the situation in the proof of lemma \ref{pq} for $p = 3$ and $q = 4x$, but nevertheless, lemma \ref{pqlem} (ii) holds for $p = 3$ and $q = 4x$. By taking $s = z$, we see that
\begin{align*}
\{3,4x\}z &= \{3y,4xy\} & \{3y,4xy\} z = \{1,2\}
\end{align*}
So $\{3z^2,4xz^2\} = \{1,2\}$. As $1z = 2$ and $4z = 3z^2 \in \{1,2\}$, $4z = 1$ follows. Furthermore, $4xz^2 = 2 = 4z^2$. As $Qz^2 = Q$, $4x = 4$ follows.

Take $s \in \Sigma$. We must show that $s$ satisfies \eqref{i}, \eqref{a}, or \eqref{b}. We distinguish three cases.
\begin{itemize}

\item $Qs \ne Q$.

In the transition graph of the power automaton, the state sets of size $n-2$ which are missing are characterized by $Q \setminus \{r,r'\}$, where $r'$ is a wildcard for the states in $Q \setminus \{1,2,3,4,r\}$. Since $|(Q \setminus \{r,r'\})x| = n-3$, the transition graph of the power automaton shows that $(Q \setminus \{1,3\})s = Q \setminus \{1,3\}$. Hence $3s = 3$. The transition graph shows that
$$
(Q \setminus \{2,4\})s \in \big\{Q \setminus \{1,3\},Q \setminus \{1,4\}\big\}
$$
as well, so $4s \in \{3,4\}$. As $3s = 3$, $4s = 4$ follows. So $s$ satisfies \eqref{b}.

\item $Qs = Q$ and $1s = 2$.

Then the transition graph of the power automaton shows that $(Q \setminus \{1,3\})s = Q \setminus \{2,4\}$. So $3s = 4$. The transition graph shows that
$$
(Q \setminus \{1,4\})s \in \big\{Q \setminus \{2,4\},Q \setminus \{1,2\}\big\}
$$
as well, so $4s \in \{4,1\}$. As $4s \ne 3s = 4$, $4s = 1$ follows. The transition graph shows that
$$
(Q \setminus \{2,4\}) s \in \big\{Q \setminus \{1,3\},Q \setminus \{1,4\}\big\}
$$
in addition, so $2s \in \{3,4\}$. As $2s \ne 3s = 4$, $2s = 3$ follows. So $s$ satisfies \eqref{a}.

\item $Qs = Q$ and $1s = 1$.

Then the transition graph of the power automaton shows that $(Q \setminus \{1,3\})s = Q \setminus \{1,3\}$. So $3s = 3$. The transition graph shows that $(Q \setminus \{2,4\})s = Q \setminus \{2,4\}$ as well. So $\{2,4\}s = \{2,4\}$. The last claim of (2) excludes the case $2s = 4$. Hence $s$ satisfies \eqref{i}. \qedhere
\end{itemize}
\end{proof}

From the proof of (2) $\Rightarrow$ (3) of theorem \ref{greedy}, we obtain the following in addition.

\begin{corollary} \label{pincor}
Suppose that $(Q,\Sigma,\ldots)$ is a Pin DFA. If $|(Qxyzx)m| = n - 2$ for all $m \in \Sigma^{*}$ such that $|m| \le 5$, then $|Qxy^3xy^3x| = n - 3$.
\end{corollary}

So $y^3x$ is a compression word for both $Qx$ and its compression $Qxy^3x$.

\end{document}